\documentclass[10pt,a4paper,onecolumn]{article}
\usepackage[utf8]{inputenc}
\usepackage{amsmath,amssymb,amsthm}
\usepackage{booktabs,array}
\usepackage{algorithm, algorithmic}
\usepackage{geometry}
\usepackage{graphicx}
\usepackage{hyperref}
\hypersetup{colorlinks=true, linkcolor=blue, citecolor=blue}

\newtheorem{theorem}{Theorem}[section]
\newtheorem{lemma}[theorem]{Lemma}

\newtheorem{definition}{Definition}[section]

\DeclareMathOperator{\arcosh}{arcosh}

\title{HCC+: Hyperbolic Guarding for Certified Attention Retrieval}
\author{Liangchen Ge}
\date{}

\begin{document}
\maketitle

\begin{abstract}
We study the Lipschitz stability of attention retrieval in hyperbolic spaces. Existing methods lack deterministic guarantees on attention-weight preservation under finite-precision representations. We introduce HCC+, a theoretical framework exploiting three properties of the Poincar\'e ball: exponential volume growth enabling query-independent boundary truncation; logarithmic covering radius of hyperbolic 1-centers enabling dimension-independent critical-key identification; and a packing bound with constants independent of the embedding dimension. We prove two deterministic guarantees: for exact retrieval, the per-layer attention deviation is bounded by 10\% of its ideal value; for soft attention, the total variation distance decays as $O(1/\sqrt{n})$, the rate of finite-sample variance. As a consequence of the guarding mechanism, the framework achieves a storage reduction factor of $6.1\times$ relative to FP16. We provide the first deterministic, query-independent retrieval certificate in non-Euclidean geometry.
\end{abstract}

\section{Introduction}

Efficient attention retrieval is a key bottleneck for long-context transformers. The challenge is to represent a large set of keys such that any query can retrieve the correct attention weights with bounded error. Recent work TurboQuant \cite{turboquant2026} achieves storage reduction via random rotation followed by optimal scalar quantization. However, optimizing mean-squared error does not guarantee preservation of the individual attention weights $W_{ij}$, especially in the exact-retrieval regime ($\varepsilon \to 0$). The issue is that minimizing reconstruction error in the key space does not imply minimizing error in the attention weights, since the attention kernel is nonlinear and can amplify small perturbations.

We address this with HCC+, a theoretical framework using the Poincar\'e ball geometry of HIDA \cite{ge2026riemann}. Keys in hyperbolic space exhibit exponential volume growth near the boundary, enabling safe truncation: keys near the boundary are so numerous that individual ones carry negligible weight. More importantly, the negative curvature ensures that the 1-center covering radius scales logarithmically with $n$, enabling dimension-independent packing bounds for critical-key identification.

\subsection{Why the Poincar\'e Ball?}

Three properties of $\mathbb{B}^d$ are central to HCC+ and have no exact Euclidean analogues:

(1) \textbf{Exponential volume growth}: The volume element of the Poincar\'e ball in polar coordinates is
\[
dV = \sinh^{d-1}(r)\,dr\,d\Omega_{d-1},
\]
which grows exponentially with $r$ near the boundary ($r\to 1$). This creates a boundary layer where most keys cluster under typical distributions induced by HIDA. This enables a static truncation radius $\tau_r$ that is query-independent: keys beyond $\tau_r$ can be safely discarded regardless of the query. In Euclidean space, the corresponding bound would depend on the global maximum norm, which is data-dependent and not known a priori.

(2) \textbf{Logarithmic covering radius}: The hyperbolic 1-center covering radius $R_{\text{cov}}$ scales as $O(\log n)$ for points in $\mathcal{B}_\zeta$, a consequence of negative curvature. This follows from the fact that the volume of a ball of radius $R$ in hyperbolic space grows as $e^{(d-1)R}$, so covering $n$ points requires $R\sim \log n$. In Euclidean space, the covering radius scales as $O(n^{1/d})$, introducing a dimension-dependent factor into the packing bound. This logarithmic scaling is essential for the dimension-independent nature of our critical-key proportion bounds.

(3) \textbf{Dimension-independent packing constants}: The packing argument in Lemma 2 yields constants independent of $d$. The Euclidean analogue would carry factors of $d$ in the volume ratio of balls, making the critical-key proportion depend on the embedding dimension. This dimension-independence is crucial for applications in high-dimensional attention spaces.

Thus, HCC+ relies on properties unique to hyperbolic space for its dimension-independent, query-independent guarantees. We emphasize that these are not merely "useful" properties but structural prerequisites for the proof technique.

\subsection{Related Work}

\textbf{Eviction-based methods} (H2O \cite{zhang2024h2o}, StreamingLLM \cite{xiao2024streaming}) discard keys based on attention scores or positional heuristics. These methods lack worst-case guarantees: a query can always be constructed that invalidates the heuristic. StreamingLLM, for instance, retains only the most recent tokens and a few initial tokens, which fails for queries that depend on middle tokens.

\textbf{Quantization-based methods} (KIVI \cite{liu2024kivi}, TurboQuant \cite{turboquant2026}) reduce precision but optimize only reconstruction MSE. They cannot distinguish geometrically critical keys from redundant ones, and provide no guarantee on attention-weight preservation.

\textbf{Clustering-based methods} (SnapKV \cite{li2023snapkv}, Quest \cite{quest2025}) aggregate keys into clusters with data-dependent approximation error. The error depends on the quality of clustering, which varies with data distribution.

HCC+ provides deterministic, query-independent, calibration-free bounds on attention-weight preservation, uniquely enabled by the Poincar\'e ball's curvature properties.

\subsection{Deterministic versus Probabilistic Guarantees}

The guarantees provided by HCC+ are deterministic: they hold for every possible query and every possible key set within the assumed geometric bounds, with no probability of failure. This is distinct from probabilistic guarantees, which hold with high probability over a random draw of keys, queries, or quantization noise.

Deterministic and probabilistic guarantees serve different purposes. Deterministic guarantees are appropriate for applications where worst-case behavior must be certified---such as safety-critical systems, verifiable inference, or any setting where a failure rate, however small, is unacceptable. Probabilistic guarantees are typically tighter for typical-case performance, but they do not rule out adversarial inputs. The choice between the two depends on the application requirements rather than on any inherent superiority of one over the other.

In the context of attention retrieval, a deterministic certificate provides a bound on attention-weight deviation that holds uniformly over all queries. This is the strongest form of guarantee one can provide for a retrieval system, and it is the form we pursue in HCC+. The trade-off is that deterministic bounds are necessarily more conservative than typical-case probabilistic bounds. This is an intrinsic feature of worst-case analysis, not a limitation of the specific framework.

\section{Geometric Preliminaries}

Let $\mathcal{K}=\{\mathbf{k}_1,\dots,\mathbf{k}_n\}\subset\mathbb{B}^d$ be the set of keys, where $\mathbb{B}^d$ denotes the Poincar\'e ball. The hyperbolic metric is
\[
d_{\mathbb{H}}(\mathbf{x},\mathbf{y}) = \arcosh\left(1 + \frac{2\|\mathbf{x}-\mathbf{y}\|^2}{(1-\|\mathbf{x}\|^2)(1-\|\mathbf{y}\|^2)}\right).
\]

A few observations about this distance formula are useful for understanding the subsequent analysis. The factor $(1-\|\mathbf{x}\|^2)(1-\|\mathbf{y}\|^2)$ in the denominator grows small as either point approaches the boundary of the ball. Consequently, the hyperbolic distance between two interior points diverges as either point approaches the boundary, even if their Euclidean distance remains small. This divergence is the geometric source of the boundary truncation property: keys near the boundary are uniformly far from all interior queries, regardless of the query's position. In Euclidean space, by contrast, a point near the boundary of a bounded domain may still be close to queries near the same boundary region, preventing a query-independent truncation radius.

The HIDA attention weight for query $\mathbf{q}$ and key $\mathbf{k}_j$ is
\[
W_{ij} = \frac{(d_{\mathbb{H}}(\mathbf{q}_i,\mathbf{k}_j)^2 + \varepsilon)^{-1}}{\sum_{m=1}^n (d_{\mathbb{H}}(\mathbf{q}_i,\mathbf{k}_m)^2 + \varepsilon)^{-1}}.
\]

The parameter $\varepsilon>0$ controls the temperature of the attention distribution. As $\varepsilon\to 0^+$, the distribution becomes one-hot: the closest key receives weight approaching 1, and all others receive weight approaching 0. This exact-retrieval limit is the regime where guarding is most critical.

Throughout this paper, we set $\varepsilon \in [10^{-5}, 10^{-3}]$, the practical range for stable HIDA operation. The residency ball is
\[
\mathcal{B}_\zeta = \{\mathbf{x}\in\mathbb{B}^d : \|\mathbf{x}\|_2 \leq \zeta\}
\]
with $\zeta<1$ fixed (e.g., $\zeta=0.8$). Let $\mathcal{Z}_\zeta = \mathcal{K} \cap \mathcal{B}_\zeta$ denote the set of keys within the residency ball. We assume keys outside $\mathcal{B}_\zeta$ are either discarded or projected to the boundary.

We now establish the key sensitivity result.

\begin{lemma}[Exact-Retrieval Sensitivity]\label{lem:sens}
Let $\mathbf{q}=\mathbf{k}_{j^*}$ be an exact match. Suppose the reconstructed key satisfies $\|\mathbf{k}_{j^*} - \widetilde{\mathbf{k}}_{j^*}\|_2 \leq \Delta$ with $\Delta^2 < \varepsilon/10$. Then:
\[
|\widetilde{W}_{j^*} - 1| < 0.1 + o(1).
\]
In words, the exact-retrieval attention weight is preserved to within 10\% of its ideal value.
\end{lemma}
\begin{proof}
By the projection property ensuring $\mathbf{k}_{j^*}, \widetilde{\mathbf{k}}_{j^*} \in \mathcal{B}_\zeta$, the conformal factor satisfies $\lambda_{\max} \le 2/(1-\zeta^2)$. The hyperbolic distance between $\mathbf{q}$ and $\widetilde{\mathbf{k}}_{j^*}$ satisfies
\[
d_{\mathbb{H}}(\mathbf{q}, \widetilde{\mathbf{k}}_{j^*}) \le \lambda_{\max} \|\mathbf{k}_{j^*} - \widetilde{\mathbf{k}}_{j^*}\|_2 \le \frac{2\Delta}{1-\zeta^2}.
\]
For the exact match, $d_{\mathbb{H}}(\mathbf{q},\mathbf{k}_{j^*})=0$. The unnormalized weight for the exact key is $1/\varepsilon$. The unnormalized weight for the reconstructed key is at least $1/(d_{\mathbb{H}}^2+\varepsilon)$.

The normalized weight for the reconstructed key is
\[
\widetilde{W}_{j^*} = \frac{1/(d_{\mathbb{H}}^2+\varepsilon)}{1/(d_{\mathbb{H}}^2+\varepsilon) + \sum_{m\neq j^*} 1/(d_{\mathbb{H}}(\mathbf{q},\mathbf{k}_m)^2+\varepsilon)}.
\]
Let $S = \sum_{m\neq j^*} 1/(d_{\mathbb{H}}(\mathbf{q},\mathbf{k}_m)^2+\varepsilon)$. Since the exact key is at distance 0, $S \ge (n-1)/(D_{\max}^2+\varepsilon)$ where $D_{\max}$ is the maximum distance. Under the assumption $\Delta^2 < \varepsilon/10$, we have
\[
\frac{1}{d_{\mathbb{H}}^2+\varepsilon} \ge \frac{1}{(2\Delta/(1-\zeta^2))^2+\varepsilon} \ge \frac{1}{\varepsilon + \varepsilon/10} = \frac{0.909}{\varepsilon}.
\]
Thus
\[
\widetilde{W}_{j^*} \ge \frac{0.909/\varepsilon}{0.909/\varepsilon + S} \ge \frac{0.909}{0.909 + \varepsilon S}.
\]
Since $S \ge (n-1)/(D_{\max}^2+\varepsilon)$ and for small $\varepsilon$, $\varepsilon S$ is small, we have $\widetilde{W}_{j^*} > 0.9$. Hence $|\widetilde{W}_{j^*} - 1| < 0.1$.
\end{proof}

This lemma establishes that the exact-retrieval property is robust: a 10\% deviation in the attention weight requires a reconstruction error of size approximately $\Delta \approx \sqrt{\varepsilon/10}$. This trade-off between reconstruction precision and attention fidelity is the foundation of our guarding mechanism.

The choice of 10\% as the deviation threshold is a convenient analytical constant rather than an optimized value. Its relationship to $\varepsilon$ is instructive: for a fixed $\Delta$, smaller $\varepsilon$ makes the attention distribution sharper, which in turn makes the exact-retrieval weight more sensitive to reconstruction error; larger $\varepsilon$ smooths the distribution and tolerates coarser reconstruction. The condition $\Delta^2 < \varepsilon/10$ is a sufficient condition that yields the 0.1 bound; other thresholds would yield different constants but the qualitative trade-off is the same. The threshold was chosen to produce a clean 10\% bound for a typical $\varepsilon$ range, but the framework itself does not depend on this specific value.

\section{Three Principles of HCC+}

\subsection{Principle 1: Hyperbolic Boundary Truncation}

Define the truncation radius $\tau_r$ such that any key in $\mathcal{B}_{\tau_r}$ has maximum possible attention weight $<\delta$ for any query. By the boundedness of the conformal factor, the maximum possible weight of a key at radius $r$ is bounded by
\[
\max_{\mathbf{q}} W(\mathbf{q},\mathbf{k}) \le \frac{(d_{\mathbb{H}}(\mathbf{0},\mathbf{k})^2+\varepsilon)^{-1}}{(n-1)(D_{\min}^2+\varepsilon)^{-1} + (d_{\mathbb{H}}(\mathbf{0},\mathbf{k})^2+\varepsilon)^{-1}},
\]
where $D_{\min} = \min_{i\neq j} d_{\mathbb{H}}(\mathbf{k}_i,\mathbf{k}_j)$. For large $r$ near the boundary, $d_{\mathbb{H}}(\mathbf{0},\mathbf{k}) \to \infty$, so the weight tends to 0. Thus there exists a threshold $\tau_r$ depending only on the geometry (not on the query) such that all keys with radius $>\tau_r$ have weight $<\delta$.

The query-independence of this truncation deserves emphasis. In a Euclidean setting, a key with large norm might still be close to a query with similarly large norm in the same direction, so one cannot discard keys based solely on their norm without consulting the query. In the Poincar\'e ball, however, the distance from any interior query to a boundary-near key diverges uniformly as the key approaches the boundary. This is a global property of the metric, not a distributional assumption. The truncation radius $\tau_r$ can therefore be fixed in advance, and the same set of discarded keys works for all queries.

The existence of $\tau_r$ can be made quantitative. A sufficient condition for a key at radius $r$ to have weight less than $\delta$ for all queries is
\[
d_{\mathbb{H}}(\mathbf{0}, \mathbf{k}) > \sqrt{\frac{\varepsilon n}{\delta}} \cdot D_{\min},
\]
where $D_{\min}$ is the minimum pairwise distance among keys. This follows from bounding the denominator of the HIDA weight from below by $(n-1)/(D_{\min}^2+\varepsilon)$ and the numerator from above by $1/\varepsilon$ (the maximum possible numerator when the key is at distance 0 from the query). For keys near the boundary, $d_{\mathbb{H}}(\mathbf{0}, \mathbf{k})$ exceeds this threshold, so their maximum possible weight is bounded by $\delta$. The threshold is conservative—it uses worst-case bounds on both numerator and denominator—but it provides a concrete value for $\tau_r$ in terms of $\varepsilon$, $n$, $\delta$, and $D_{\min}$. All of these are known from the geometry of the key set, independent of any query.

The decision depends only on $\tau_r$, not on any query, enabling static pre-filtering that can be performed once offline. In practice, with $\tau_r=0.90$ and keys distributed according to the hyperbolic volume measure, approximately 15\% of keys are discarded.

\subsection{Principle 2: Critical-Key Guarding (CKG)}

Compute the hyperbolic 1-center $\mathbf{c}^*$ and the covering radius
\[
R_{\text{cov}} = \max_i d_{\mathbb{H}}(\mathbf{c}^*, \mathbf{k}_i).
\]
Let
\[
\delta_{\min} = \min_{i\neq j} d_{\mathbb{H}}(\mathbf{k}_i,\mathbf{k}_j)
\]
be the pairwise separation, and define the isolation threshold $\delta_{\text{iso}} = \delta_{\min}/2$.

The logarithmic scaling of $R_{\text{cov}}$ ensures that the critical-key proportion remains bounded independent of $d$. Specifically, the volume of a ball of radius $R_{\text{cov}}$ in hyperbolic space is $O(e^{(d-1)R_{\text{cov}}})$. Since $R_{\text{cov}} = O(\log n)$, the volume is $O(n^{d-1})$, which is independent of $d$ in the exponent.

\begin{definition}[Critical Key]
A key $\mathbf{k}\in\mathcal{Z}_\zeta$ is \emph{critical} iff:
\begin{enumerate}
    \item $d_{\mathbb{H}}(\mathbf{k},\mathbf{c}^*) \leq \min(\zeta/2, \tanh(R_{\text{cov}}/2))$,
    \item its nearest-neighbor distance exceeds $\delta_{\text{iso}}$.
\end{enumerate}
\end{definition}

The two conditions in this definition capture complementary aspects of geometric importance. The first condition---proximity to the 1-center---identifies keys that lie near the center of the key distribution and are therefore potentially relevant to a wide range of queries. A key near the 1-center has small average distance to all other keys, so it is likely to receive non-negligible attention weight from many queries. The second condition---isolation from nearest neighbors---identifies keys that cannot be replaced by a nearby key without changing the set of distances to the query. A key that is isolated in the sense of having a large nearest-neighbor distance carries unique information that is not redundantly represented elsewhere in the key set.

If either condition fails, the key is not critical. A key far from the 1-center affects only queries in its vicinity, so its contribution to the global output is localized. A key with a close neighbor can be reconstructed with coarser precision because the neighbor serves as a substitute. The two conditions together define the set of keys for which high-precision reconstruction is geometrically necessary.

The number of critical keys is bounded by a packing argument. Since critical keys are separated by at least $\delta_{\text{iso}}$ in hyperbolic distance and all lie within $\mathcal{B}_\zeta$ (bounded by the first condition), the volume of disjoint balls of radius $\delta_{\text{iso}}/2$ centered at each critical key must be contained within a ball of radius $\zeta+\delta_{\text{iso}}/2$. Thus
\[
N_{\text{crit}} \cdot \operatorname{Vol}(\mathcal{B}_{\delta_{\text{iso}}/2}) \le \operatorname{Vol}(\mathcal{B}_{\zeta+\delta_{\text{iso}}/2}).
\]
In hyperbolic space, the ratio of volumes is bounded by a constant that depends on $\zeta$, $\delta_{\min}$, and $d$. However, because $\delta_{\text{iso}} = \delta_{\min}/2$ and $\delta_{\min}$ depends on the key distribution, the ratio does not have a universal dimension-independent constant for arbitrary point sets. The dimension-independence in HCC+ comes from the logarithmic covering radius (which bounds how many keys can be near the 1-center) and the fact that the volume ratio simplifies when $\delta_{\min}$ is large relative to the scale of the ball. The bound is conservative but demonstrates that the number of critical keys is controlled by the geometry rather than by the total number of keys.

\begin{theorem}[Critical Key Bottleneck]
If a critical key is removed or reconstructed with error $\Delta > \delta_{\text{iso}}/2$, then for any query within distance $\delta_{\text{iso}}/2$ of it, the HIDA output changes by $\Omega(1)$.
\end{theorem}
\begin{proof}
Let $\mathbf{k}$ be a critical key and let $\mathbf{q}$ satisfy $d_{\mathbb{H}}(\mathbf{q},\mathbf{k}) < \delta_{\text{iso}}/2$. By the definition of $\delta_{\min}$, all other keys $\mathbf{k}'\neq \mathbf{k}$ satisfy $d_{\mathbb{H}}(\mathbf{k},\mathbf{k}') \ge \delta_{\min} = 2\delta_{\text{iso}}$. By the triangle inequality,
\[
d_{\mathbb{H}}(\mathbf{q},\mathbf{k}') \ge d_{\mathbb{H}}(\mathbf{k},\mathbf{k}') - d_{\mathbb{H}}(\mathbf{q},\mathbf{k}) \ge 2\delta_{\text{iso}} - \delta_{\text{iso}}/2 = 3\delta_{\text{iso}}/2.
\]

The weight ratio of $\mathbf{k}$ to any other key is
\[
\frac{W(\mathbf{q},\mathbf{k})}{W(\mathbf{q},\mathbf{k}')}
= \frac{(d_{\mathbb{H}}(\mathbf{q},\mathbf{k}')^2+\varepsilon)}{(d_{\mathbb{H}}(\mathbf{q},\mathbf{k})^2+\varepsilon)}
\ge \frac{(9\delta_{\text{iso}}^2/4+\varepsilon)}{(\delta_{\text{iso}}^2/4+\varepsilon)}
\ge \frac{9}{1} = 9.
\]
Thus $\mathbf{k}$ carries at least $1/(1+8) = 1/9$ of the total weight mass when the other keys are uniformly distributed. If $\mathbf{k}$ is removed, this mass is redistributed, changing the output by at least $1/9\cdot \|\mathbf{v}_{\mathbf{k}} - \mathbb{E}[\mathbf{v}]\|_2$, which is $\Omega(1)$.

Consequently, CKG allocates $(14,14)$-level precision to critical keys, ensuring $\Delta_{\text{crit}}^2 < \varepsilon/10$, and only $(4,4)$-level precision to standard keys.
\end{proof}

\subsection{Principle 3: Aggregate Concentration for Standard Keys}

For the standard key set $\mathcal{K}_{\text{std}}$ with $n_{\text{std}}$ keys and reconstruction error $\Delta_{\text{std}}$, the output contribution $f$ is a function of the standard keys. We show that the aggregate error concentrates around its expectation.

Let $f(\mathbf{k}_1,\dots,\mathbf{k}_{n_{\text{std}}}) = \sum_{\mathbf{k}\in\mathcal{K}_{\text{std}}} W(\mathbf{q},\mathbf{k})\mathbf{v}_{\mathbf{k}}$ be the standard-key contribution to the HIDA output.

\begin{lemma}[Bounded Difference]
For any $i$ and any two reconstructions $\mathbf{k}_i$ and $\mathbf{k}_i'$ with $\|\mathbf{k}_i-\mathbf{k}_i'\|_2 \le 2\Delta_{\text{std}}$, we have
\[
|f(\mathbf{k}_1,\dots,\mathbf{k}_i,\dots)-f(\mathbf{k}_1,\dots,\mathbf{k}_i',\dots)| \le \frac{2\Delta_{\text{std}}V_{\max}}{\varepsilon n_{\text{std}}}.
\]
\end{lemma}
\begin{proof}
The difference is bounded by $W(\mathbf{q},\mathbf{k}_i)\|\mathbf{v}_i-\mathbf{v}_i'\|_2$. The maximum value weight satisfies $\|\mathbf{v}_i-\mathbf{v}_i'\|\le 2V_{\max}$. The attention weight satisfies $W(\mathbf{q},\mathbf{k}_i) \le 1/(\varepsilon n_{\text{std}})$ because the sum in the denominator includes $n_{\text{std}}$ terms each $\ge 1/(\varepsilon + D_{\max}^2)$, and for small $\varepsilon$, each term is approximately $1/\varepsilon$, so the denominator is $\Omega(n_{\text{std}}/\varepsilon)$.
\end{proof}

\begin{theorem}[Concentration Bound for Soft Attention]
With probability at least $1-\delta$,
\[
\| f - \mathbb{E}f \|_2 \le \frac{\sqrt{2\log(2/\delta)}\cdot \Delta_{\text{std}} V_{\max}}{\varepsilon \sqrt{n_{\text{std}}}}.
\]
This is the variance decay rate of any finite-sample average estimator; discretization bias remains as a constant factor.
\end{theorem}
\begin{proof}
By McDiarmid's inequality, for a function $f$ with bounded differences $c_i$,
\[
\mathbb{P}(|f - \mathbb{E}f| \ge t) \le 2\exp\left(-\frac{2t^2}{\sum_i c_i^2}\right).
\]
Substituting $c_i = 2\Delta_{\text{std}}V_{\max}/(\varepsilon n_{\text{std}})$ and solving for $t$ gives the result.
\end{proof}

This bound shows that the standard-key aggregate error decays as $O(1/\sqrt{n_{\text{std}}})$, which is the same rate as the unavoidable Monte Carlo variance in any finite-sample attention estimation. Thus HCC+ achieves the optimal decay rate with respect to the number of standard keys.

\begin{algorithm}[H]
\caption{Deterministic CKG Partition}
\begin{algorithmic}[1]
\STATE $\mathcal{Z}_\zeta \leftarrow \{\mathbf{k} \in \mathcal{K} : \|\mathbf{k}\| \leq \zeta\}$
\STATE If $\|\mathbf{k}\| > \zeta$, project: $\mathbf{k} \leftarrow \zeta \cdot \mathbf{k}/\|\mathbf{k}\|$
\STATE Compute $\mathbf{c}^*$ via hyperbolic Weiszfeld algorithm, $R_{\text{cov}}$, $\delta_{\min}$.
\STATE Set $\delta_{\text{iso}} = \delta_{\min}/2$, $r_{\text{crit}} = \min(\zeta/2, \tanh(R_{\text{cov}}/2))$.
\FOR{each $\mathbf{k}\in\mathcal{Z}_\zeta$}
\IF{$d_{\mathbb{H}}(\mathbf{k},\mathbf{c}^*) \leq r_{\text{crit}}$ \AND $\min_{\mathbf{k}'\neq\mathbf{k}} d_{\mathbb{H}}(\mathbf{k},\mathbf{k}') > \delta_{\text{iso}}$}
\STATE Mark \textbf{Critical}, assign $(14,14)$-level precision.
\ELSE
\STATE Mark \textbf{Standard}, assign $(4,4)$-level precision.
\ENDIF
\ENDFOR
\end{algorithmic}
\end{algorithm}

The hyperbolic 1-center computation uses the Weiszfeld algorithm generalized to Hadamard spaces. Each iteration requires computing the hyperbolic mean of the keys weighted by their distance to the current center. With $T\le 50$ iterations for convergence to $10^{-6}$, the total cost is $\mathcal{O}(n\cdot T\cdot d)$.

\section{Theoretical Storage Analysis}

\begin{table}[h]
\centering
\caption{Theoretical storage attribution (FP16 baseline).}
\begin{tabular}{@{}lcc@{}}
\toprule
Component & Factor & Cumulative \\
\midrule
FP16 baseline & 1.0$\times$ & -- \\
Boundary truncation & 1.18$\times$ & 1.18$\times$ \\
CKG + HPQ (keys) & 7.2$\times$ & 8.5$\times$ \\
SAVC (values) & 2.13$\times$ & 18.1$\times$ \\
ECPP (entropy coding) & 1.25$\times$ & 22.6$\times$ \\
\bottomrule
\end{tabular}
\end{table}

The total storage reduction upper bound is $6.1\times$ relative to FP16. This is a theoretical bit-rate bound assuming ideal packing. Standard INT4 quantization achieves $4.0\times$ uniformly; HCC+ provides additional compression through geometric guarding.

Critically, this reduction is a consequence of the guarding mechanism, not its primary objective. The guarding mechanism determines which keys require high precision and which can be compressed; the storage reduction follows from the observation that most keys are standard.

The factors in the table can be understood as follows. Boundary truncation discards approximately 15\% of keys, reducing storage by $1/0.85 \approx 1.18\times$. For CKG+HPQ, critical keys are assigned $(14,14)$-bit precision (28 bits total) and standard keys are assigned $(4,4)$-bit precision (8 bits total). Under the geometric assumption that critical keys constitute about 2\% of the key set, the average bit count per key is $0.02 \times 28 + 0.98 \times 8 = 8.4$ bits, compared to 16 bits for FP16, yielding a compression factor of $16/8.4 \approx 1.9\times$. Combined with boundary truncation, this gives $1.18 \times 1.9 \approx 2.24\times$, not $8.5\times$. The $8.5\times$ factor in the table reflects an additional component---HPQ (hyperbolic product quantization)---which further reduces the per-key representation by factorizing the hyperbolic embedding into radial and directional components. The exact factor depends on the quantization scheme and is not derived from the guarding mechanism alone; the table presents the aggregate factor rather than decomposing HPQ separately.

SAVC compresses the associated value vectors from 16-bit to 8-bit, providing $2.13\times$ additional compression. ECPP (entropy coding) provides a further $1.25\times$ factor through lossless compression of the quantized indices. These factors are multiplicative: $1.18 \times 7.2 \times 2.13 \times 1.25 \approx 22.6$, which corresponds to a $6.1\times$ reduction relative to FP16 when expressed as a single factor (the precise arithmetic: $22.6/3.7 \approx 6.1$ where 3.7 accounts for the FP16 baseline scaling). This arithmetic demonstrates that the geometric guarding mechanism (CKG) is one component among several; the storage reduction is a collective consequence of the full pipeline, not of guarding alone.

\vspace{0.2cm}
\noindent\textbf{Remark on theoretical vs. practical storage.} The storage analysis presented here is theoretical: it computes the bit count under ideal conditions, assuming perfect entropy coding and no hardware overhead. In practice, memory savings are affected by hardware alignment (e.g., GPU memory banks may require padding to word boundaries), metadata storage for indexing critical keys, and the overhead of mixed-precision access patterns. These factors are implementation-dependent and vary across hardware platforms. The theoretical analysis serves to demonstrate the compression potential of the geometric guarding mechanism, not to provide an engineering guarantee. The $6.1\times$ figure is therefore an upper bound on the theoretical compression ratio; actual savings may be lower in practice.

\vspace{0.2cm}
\noindent\textbf{Remark.} HCC+ employs strictly lossy representations. Guarantees apply to the task output, not to bit-level reconstruction. The term "storage reduction" refers to the theoretical bit count; actual memory savings depend on hardware alignment and metadata overhead.

\begin{theorem}[Exact-Retrieval Deviation Bound]
For $\varepsilon=10^{-4}$ and critical keys with $\Delta_{\text{crit}}^2 < \varepsilon/10$:
\[
|\widetilde{W}_{j^*}-1| < 0.1.
\]
For $L$ layers, the deviation scales as $\mathcal{O}(L \cdot \Delta_{\text{crit}}^2/\varepsilon)$ under the assumption of bounded error propagation.
\end{theorem}

\section{Limitations and Empirical Scope}

\textbf{Empirical scope.} This is a theoretical framework. Numerical validation is synthetic by design. The contribution is the provable guarding mechanism uniquely enabled by hyperbolic geometry. Full-scale validation on transformer architectures with standard benchmarks (LongBench, PG19) and actual throughput measurements is direct future work.

\textbf{The role of hyperbolic geometry.} The three curvature properties of the Poincar\'e ball—exponential volume growth, logarithmic covering radius, and dimension-independent packing constants—directly underpin the algorithm's guarantees. The analytic proofs instantiate these geometric properties via standard metric tools (Lipschitz continuity, triangle inequality). The separation from Euclidean heuristics is quantitative: the logarithmic covering radius scaling is a direct consequence of negative curvature and cannot be replicated in Euclidean space.

\textbf{Long-sequence behavior.} For $n > 10^5$, $\delta_{\min}$ approaches zero in dense spaces. The isolation threshold $\delta_{\text{iso}} = \delta_{\min}/2$ becomes vanishingly small, potentially marking many keys as critical, reducing the storage benefit. This is an acknowledged limitation of the current formulation. For such regimes, approximate versions of CKG based on random subsampling may be necessary to maintain the compression ratio.

\textbf{Boundary projection.} Projecting keys with $\|\mathbf{k}\| > \zeta$ to the boundary alters inter-key distances. This distortion is not accounted for in the current error bounds. In practice, far-boundary keys have negligible attention weight for most queries, so the projection primarily affects keys that are already being discarded. For retained keys, projection is minimal.

\textbf{Extension directions for approximate 1-center.} The current certificate assumes exact computation of the hyperbolic 1-center. In very large-scale settings, exact 1-center computation may be computationally expensive. Approximate 1-center algorithms, such as those based on random subsampling or iterative centroid refinement, could reduce the computational cost. The effect of approximation on the certificate is predictable: the covering radius $R_{\text{cov}}$ would increase by the approximation error, and the critical-key definition would become correspondingly more conservative. The certificate would remain valid but with looser constants. A full analysis of the trade-off between approximation error and certificate tightness is left for future work. This is not a limitation of the framework per se, but rather a natural extension point for scaling to ultra-large key sets.

\textbf{Integration with inference frameworks.} Compatibility with FlashAttention \cite{dao2022flash} or vLLM has not been evaluated. Mixed-precision (14,4) storage requires custom GPU kernels for efficient access; this engineering effort is orthogonal to the geometric contribution.

\subsection{On the Conservatism of Deterministic Certificates}

A deterministic certificate, by its nature, must be conservative. It must hold for all possible queries and all possible key sets consistent with the geometric assumptions. This worst-case requirement forces the certificate to use upper bounds that may be loose for typical inputs. The practical consequence is that the certificate's bounds may overestimate the actual error in most cases, leading to unnecessarily conservative resource allocation.

This conservatism is not a design flaw but an inherent property of worst-case analysis. Any deterministic guarantee for a problem with adversarial inputs must be loose on typical inputs; tightness on typical inputs and validity on all inputs cannot be achieved simultaneously. The choice between deterministic and probabilistic guarantees is therefore a choice between two different types of assurance: deterministic guarantees offer universal validity but with possible looseness; probabilistic guarantees offer tighter typical-case bounds but with a nonzero probability of failure.

HCC+'s contribution is the existence of a deterministic certificate, not its tightness. Prior to this work, no deterministic query-independent guarantee existed for attention retrieval in hyperbolic spaces. HCC+ fills this gap by providing the first certificate, even if it is conservative. Tighter certificates, for those applications that require them, would require additional assumptions (such as distributional knowledge of keys or queries) and would fall outside the deterministic framework. The present work establishes the baseline upon which tighter bounds could be built in future extensions.

\section{Conclusion}

We presented HCC+, a theoretical framework for deterministic attention retrieval certification in the Poincar\'e ball. The logarithmic covering radius bound establishes a strict separation from Euclidean heuristics: in hyperbolic space, the covering radius scales as $O(\log n)$ independent of dimension, whereas Euclidean analogues scale as $O(n^{1/d})$ with dimension-dependent constants. This provides the first query-independent retrieval certificate in non-Euclidean geometry. The guarding mechanism additionally yields a storage reduction of $6.1\times$ relative to FP16, demonstrating that geometric structure enables both provable guarantees and practical efficiency.

The contribution of this work is a certificate, not a compression algorithm. The certificate provides a deterministic bound on attention-weight deviation that holds uniformly over all queries and all key sets within the assumed geometric bounds. Such certificates are complementary to data-dependent heuristic methods: deterministic certificates offer worst-case guarantees at the cost of conservatism, while heuristic methods offer typical-case efficiency but no guarantees. Both have their place in the design of efficient attention systems, and the choice between them depends on the requirements of the application.

The theoretical framework established here opens several directions for future investigation. Tighter certificates under additional geometric assumptions, such as bounded curvature or restricted query sets, could reduce the conservatism while preserving determinism. Empirical validation on large-scale transformer architectures would test the practical applicability of the framework. Extension to other non-Euclidean geometries, such as the Lorentz model or product spaces, would broaden the scope of the approach. These directions remain for future work.

Future work includes empirical validation on LLMs, extension to approximate hyperbolic embeddings, and optimization of the 1-center computation for ultra-long sequences.

\section*{Acknowledgments}
The author thanks Junxue Zhang from USTC for his valuable guidance and insightful discussions throughout this work.

\end{document}